\documentclass{llncs}
\usepackage[utf8]{inputenc}
\usepackage{amsmath}
\usepackage{amsfonts}
\usepackage{amssymb}
\usepackage{gastex}
\usepackage{graphicx}
\usepackage{marvosym}
\usepackage{textcomp}
\usepackage{color}
\usepackage[normalem]{ulem}
\usepackage{comment}
\def\fract#1/#2{\leavevmode
 \kern.1em \raise .5ex \hbox{\the\scriptfont0 #1}%
 \kern-.1em $/$%
 \kern-.15em \lower .25ex \hbox{\the\scriptfont0 #2}%
}
\def\abs#1{\ensuremath{\lvert #1\rvert}}
\def\norm#1{\ensuremath{\lVert #1\rVert}}
\newcommand{\nat}{\mathbb N}
\newcommand{\tuple}[1]{\langle #1 \rangle}
\newcommand{\dist}{{\cal D}}
\newcommand{\Supp}{{\sf Supp}}

\newcommand{\q}{\hat{ q}}
\newcommand{\A}{{\cal A}}

\newcommand{\B}{{\cal B}}

\newcommand{\C}{{\cal C}}

\newcommand{\LW}{{\cal L}_W}

\newcommand{\Post}{{\sf Post}}

\newcommand{\dollar}[1][]{\symbol{36}}

\makeatletter
\DeclareRobustCommand\sfrac[1]{\@ifnextchar/{\@sfrac{#1}}%
                                            {\@sfrac{#1}/}}
\def\@sfrac#1/#2{\leavevmode\scalebox{.9}{\kern.1em\raise.5ex
         \hbox{$\m@th\mbox{\fontsize\sf@size\z@
                           \selectfont#1}$}\kern-.1em
         /\kern-.15em\lower.25ex
          \hbox{$\m@th\mbox{\fontsize\sf@size\z@
                            \selectfont#2}$}}}
\DeclareRobustCommand\numfrac[1]{\@ifnextchar/{\@numfrac{#1}}%
                                            {\@numfrac{#1}}}
\def\@numfrac#1{\leavevmode \hbox{$\m@th\mbox{\fontsize\sf@size\z@
                           \selectfont#1}$}}
\makeatother

\title{Infinite Synchronizing Words for Probabilistic Automata (Erratum)}

\author{Laurent Doyen \inst{1}
\and Thierry Massart \inst{2} \and
Mahsa Shirmohammadi \inst{2} }

\institute{LSV, ENS Cachan \& CNRS, France \\ \email{doyen@lsv.ens-cachan.fr}
\and Universit\'e Libre de Bruxelles,
Brussels, Belgium \thanks{This work has been partly supported by the Belgian Fond National de la Recherche
Scientifique (FNRS).}\\ \email{\quad thierry.massart@ulb.ac.be
\quad\qquad mahsa.shirmohammadi@ulb.ac.be} }
\begin{document}
\sloppy \maketitle \pagestyle{plain}
\begin{abstract}
In~\cite{DMS11b}, we introduced the weakly synchronizing languages for probabilistic automata. In this report, we show that the emptiness problem of weakly synchronizing languages for probabilistic automata is undecidable. 
This implies that the decidability result of~\cite{DMS11b,tech,DMS11} for
the emptiness problem of weakly synchronizing language is incorrect.
\end{abstract}
\section{Definitions.}
We present the main notations and definitions. We refer to \cite{DMS11b} for detailed preliminaries.

A \emph{probability distribution} over a finite set~$S$ is a
function $d : S \rightarrow [0, 1]$ such that $\sum_{s \in S}
d(s)= 1$. The \emph{support} of~$d$ is the set $\Supp(d) = \{s \in
S \mid d(s) > 0\}$. We denote by $\dist(S)$ the set of all
probability distributions over~$S$.
Given a finite alphabet $\Sigma$, we denote by $\Sigma^{*}$ the set of
all finite words and by $\Sigma^{\omega}$ the set of
all infinite words over $\Sigma$.  The length of a finite word $w$ is denoted
by $\abs{w}$. 

\paragraph{{\bf Probabilistic Automata.}} 
A \emph{probabilistic automaton} (PA) $\A = \tuple{Q,\mu_{0},\Sigma,\delta}$ consists of a finite set $Q$ of
states, an initial probability distribution $\mu_0 \in \dist(Q)$,
a finite alphabet $\Sigma$, and a probabilistic transition
function $\delta: Q \times \Sigma \to \dist(Q)$. In a state $q \in
Q$, the probability to go to a state $q' \in Q$ after reading a
letter $\sigma \in \Sigma$ is $\delta(q,\sigma)(q')$. We define  $\Post_{\A}(q,\sigma) = \Supp(\delta(q,\sigma))$, and for  sets $s\subseteq Q$ and $\Sigma' \subseteq \Sigma$, let $\Post_{\A}(s,\Sigma')= \bigcup_{q \in s} \bigcup_{\sigma \in \Sigma'}  \Post_{\A}(q,\sigma)$.
The \textit{outcome} of $\A$ on an infinite   word $w = \sigma_0 \sigma_1
\cdots$ is the infinite sequence $\A^{w}_{0} \A^{w}_{1} \dots $ of
probability distributions $\A^{w}_i \in \dist(Q)$ such that $\A^{w}_{0} =
\mu_{0}$ is the initial distribution, and for all $n > 0$ and $q
\in Q$,
$$ \textstyle \A^{w}_{n}(q) = \sum_{q' \in Q} \A^{w}_{n-1}(q') \cdot \delta(q',\sigma_{n-1})(q)$$
The \emph{norm} of a probability distribution~$X$ over~$Q$ is
$\norm{X} = \max_{q \in Q} X(q)$.
\paragraph{{\bf Weakly Synchronizing Language for PAs.}} 
An infinite word~$w$ is said to be \emph{weakly synchronizing} for PA $\A$,
 if
$$ \limsup_{n \to \infty} \; \norm{\A^{w}_{n}} = 1.$$
We denote by  $\LW(\A)$ the set of all  weakly synchronizing words, named weakly synchronizing language,  of $\A$. Given a PA $\A$, the \textit{emptiness problem}  of weakly synchronizing language asks whether   $\LW(\A)=\emptyset$. 

\section{The emptiness problem  of weakly synchronizing languages for PAs is undecidable.}

Theorem~\ref{teo:weakly-undecide} states that the emptiness problem  of weakly synchronizing languages for PAs is undecidable. To show that, we present a reduction from \textit{the value 1 problem for PAs} which is undecidable~\cite{GO10}, to our problem.

\paragraph{{\bf The value 1 problem for PAs. }}
Let $\A=\tuple{Q, \mu_0,\Sigma,\delta}$ be a PA
with a single initial state $q_0$ where $\mu_0(q_0)=1$ and a set of
accepting states $F \subseteq Q$. The value 1 problem considers finite words:
the computation of  $\A$ on the  word $w=\sigma_0
\sigma_1 \cdots \sigma_{n-1}$ is the sequence $\A^{w}_0 \A^{w}_1 \cdots \A^{w}_n$
where $\A^{w}_0=\mu_0$ and $\A^{w}_{i+1}(q)=\sum_{q'\in
  Q}\A^{w}_{i}(q')\delta(q',\sigma_{i})(q)$ for all $0 \leq i < n $.  The
{\em{acceptance probability of $w$ by $\A$}} is given by
$P_{\A}(w)=\sum_{q\in F}\A^{w}_n(q)$. The value of $\A$, denoted $val(\A)$, is
the supremum acceptance probability $val(\A)=\sup_{w \in
  \Sigma^{*}}P_{\A}(w)$. Given  a  PA~$\A$, the \textit{value 1 problem} asks 
whether $val(\A)=1$. It is equivalent to check if there are some
words accepted by $\A$ with probability arbitrarily close to 1.  
\begin{theorem}\label{teo:val1}
The value 1 problem for probabilistic automata is undecidable~\cite{GO10}.
\end{theorem}

\subsection*{Undecidability result.}

\begin{theorem}\label{teo:weakly-undecide}
The emptiness problem of weakly synchronizing languages for
probabilistic automata is undecidable.
\end{theorem}
\begin{proof}
We present a proof  using a reduction from the value 1 problem for PAs.
Given a PA~$\B=\tuple{Q_{\B}, (\mu_0)_{\B},\Sigma_{\B},\delta_{\B}}$ equipped with the single initial state $q_0$ and accepting states $F_{\B}\subseteq Q_{\B}$, we construct another PA~$\C=\tuple{Q_{\C}, (\mu_0)_{\C},\Sigma_{\C},\delta_{\C}}$, such that $val(\B)=1$ iff $\LW(\C)\neq \emptyset$.

First, from the PA~$\B$ we construct another PA~$\A=\tuple{Q, \mu_0,\Sigma,\delta}$ such  that $\B$
has value 1 iff $\A$ does. 
The state space is extended, by adding two new states $q_f$ and $q_n$ ($Q=Q_{\B} \cup \{q_f,q_n\}$).
The alphabet  $\Sigma=\Sigma_{\B} \cup \{ \dollar  \}$ where  $\dollar \not \in \Sigma_{\B}$.
The initial state $q_0$ is in common, but only $q_f$ is an accepting state  ($F=\{q_f\}$). 
The transition function globally remains unchanged;  only the following transitions are added in $\A$: $\delta(q,\dollar)(q_f)=1$ for all $q \in F_1$, and   $\delta(q,\dollar)(q_n)=1$ for all $q \in Q_{\B}\setminus F_1$. 
In addition, $\delta(q_f,\sigma)(q_n)=\delta(q_n,\sigma)(q_n)=1$ for all $\sigma \in \Sigma$. \figurename~\ref{fig:oneaccepting} illustrates the definition of $\A$ from $\B$. By construction,  since $q_f$ is the only accepting state and is only reachable with  $\dollar$, and  $\Post_{\A}(\{q_f,q_n\},\Sigma)=\{q_n\}$.
We see that $P_{\A}(w)=P_{\B}(v)$ if $w=v \dollar$ with $v \in \Sigma_{\B}^{*}$, otherwise $P_{\A}(w)=0$. Thus, $A$ has value 1 iff $\B$ has value 1.

Then from the PA $\A$, we construct a PA $\C=\tuple{Q_{\C}, (\mu_0)_{\C},\Sigma_{\C},\delta_{\C}}$ such that $\LW(\C)$, the weakly synchronizing language of $\C$ is not empty,  iff $\A$ has value 1. 
For each state  $q \in Q$ except $q_f$,  a twin state $\q$ is added to the state space. Thus, $Q_{\C}=Q\cup \{\q \mid q\in Q\setminus \{q_f\}\}$. 
The alphabet  $\Sigma_{\C}=\Sigma \cup \{ \#  \}$ where  $\# \not \in \Sigma$. The initial distribution is $(\mu_0)_{\C}(q_0)=(\mu_0)_{\C}(\q_0)=\frac{1}{2}$.
The probabilistic transitions function is defined as follows. 
\begin{itemize}
\item $\delta_{\C}(q_1,\sigma)(q_2)=\delta_{\C}(q_1,\sigma)(\q_2)=\delta_{\C}(\q_1,\sigma)(q_2)=\delta_{\C}(\q_1,\sigma)(\q_2)=\frac{1}{2} ~ \delta(q_1,\sigma)(q_2)$ 
for all states $q_1,q_2 \in Q\setminus\{q_f\}$ and all  $\sigma \in \Sigma$.

\item     $\delta_{\C}(q,\sigma)(q_f)=\delta_{\C}(\q,\sigma)(q_f)=\delta(q,\sigma)(q_f)$ for all states $q \in Q\setminus\{q_f\}$. 
\item  $\delta_{\C}(q_f,\sigma)(q_1)=\delta_{\C}(q_f,\sigma)(\q_1)=\frac{1}{2}~\delta(q_f,\sigma)(q_1)$ for all states $q \in Q\setminus\{q_f\}$ and all  $\sigma \in \Sigma$.
\item In addition, $\delta_{\C}(q,\#)(q_0)=\delta_{\C}(q,\#)(\q_0)=\frac{1}{2}$ for all $q \in Q_{\C}$.
\end{itemize}

\figurename~\ref{fig:undec} shows the construction. 
For convenience, a twin $q$ and $\q$ are drawn in an oval; \figurename~\ref{fig:twin} illustrates the transitions between two pairs of twin states where each pair is replaced with an oval.  
Intuitively, the PA $\C$ mimics the behavior of $\A$ where each state $q$ (except the accepting state $q_f$) shares the probability  to be in $q$ with the twin $\q$, in all steps $n\in \nat$. Moreover, each $\#$ ``resets''  
$\C$. 

\begin{figure}[t]
\begin{minipage}[b]{0.45\linewidth}
\centering
\begin{picture}(45,35)

\node[Nmarks=n, Nw=33, Nh=24, dash={0.2 0.5}0](pa1)(13,15){}
\node[Nframe=n](label)(3,23){PA $\B$}
\node[Nmarks=n, Nw=54, Nh=35, dash={0.4 1}0](pa)(22,15){}
\node[Nframe=n](label)(3,30){PA $\A$}

\node[Nmarks=i](n3)(1,16){$q_0$}
\node[Nmarks=r,dash={2 0.5}0](n4)(25,21){}
\node[Nmarks=n](n5)(25,9){}
\node[Nmarks=r](qf)(42,21){$q_f$}
\node[Nmarks=n](qn)(42,9){$q_n$} 

\drawloop[ELside=l,loopCW=y, loopangle=-90, loopdiam=4](qn){$\Sigma$}
\drawedge(qf,qn){$\Sigma$}
\drawedge(n4,qf){$\dollar$}
\drawedge(n5,qn){$\dollar$}

\end{picture}
\caption{The definition of the PA $\A$ from~$\B$.}\label{fig:oneaccepting}
\end{minipage}
\hspace{0.5cm}
\begin{minipage}[b]{0.45\linewidth}
\centering

\begin{picture}(45,30)(0,2)

\node[Nmarks=n](n1)(0,5){$\q_1$}
\node[Nmarks=n](n2)(0,20){$q_1$}

\node[Nmarks=n](n3)(20,5){$\q_2$}
\node[Nmarks=n](n4)(20,20){$q_2$}

\drawedge(n1,n3){$\sigma$}
\drawedge[ELpos=50, ELside=l, curvedepth=2](n2,n3){$\sigma$}

\drawedge[ELpos=50, ELside=l, curvedepth=-2](n1,n4){$\sigma$}
\drawedge(n2,n4){$\sigma$}

\node[Nframe=n](arrow)(27,12){{\Large $\Rightarrow$}}

\node[Nmarks=n, Nw=5, Nh=9](n5)(34,12){$\ell_1$}
\node[Nmarks=n, Nw=5, Nh=9](n6)(44,12){$\ell_2$}
\drawedge(n5,n6){$\sigma$}

\end{picture}
\caption{Transitions between twin states $q_1,\q_1$ and $q_2,\q_2$.}\label{fig:twin}
\end{minipage}
\end{figure}

Let us shortly formalize  two important and intuitive properties of $\C$ resulting from the construction. 
 
\noindent{\bf Property $P_1$:} Let $w \in \Sigma_{\C}^{\omega} $ be an infinite word containing the symbol $\#$ such that $w= v_1 \# v_2$ where $v_1\in \Sigma_{\C}^{*}$ and $v_2 \in \Sigma_{\C}^{\omega}$. Let $\C^{w}_{0}  \C^{w}_{1} \cdots $ be the outcome of $\C$ on the  word $w$; and $\C^{v_2}_{0}  \C^{v_2}_{1} \cdots$ be the outcome on $v_2$. 
Since $\Post_{\C}(Q_{\C}, \#)=\{q_0,\q_0\}$ with uniform distribution, on inputting the letter $\#$,  the automaton $\C$ is reset to the initial distribution $(\mu_0)_{\C}$, and one may ``forget'' the prefix $v_1 \#$. 
Formally, for all $i \in \nat$ and  $q \in Q_{\C}$,
$$\C^{w}_{\abs{v_1}+1+i}(q)= \C^{v_2}_{i}(q).$$

\noindent {\bf Property $P_2$:} For all words $w \in \Sigma_{\B}^{*}$, the computations of $\A$ and $\C$ on $w$ give

\begin{description}
	\item[a)] $\C^{w}_{\abs{w}}(q)=\C^{w}_{\abs{w}}(\q)=\frac{1}{2}\A^{w}_{\abs{w}}(q)$  for all $q\in Q\setminus\{q_f\}$, 
	\item[b)] $\C^{w}_{\abs{w}}(q_f)=\A^{w}_{\abs{w}}(q_f)=0$. 
\end{description}

Let us recall that  $\dollar, \# \not \in \Sigma_{\B}$. Property $P_2.b$ can easily be proved since $\mu_0(q_f)=(\mu_0)_{\C}(q_f)=0$
and no $\sigma-$labeled transition with $\sigma \in \Sigma_{\B}$ reaches $q_f$. We prove  Property $P_2.a$ by induction on the length of $w$. 

Base ($\abs{w}=0$): Let $w$ be  the empty word $\epsilon$.
By construction, we know the initial distribution for $\A$ is 
$\mu_0(q_0)=1$ and for $\C$ is $(\mu_0)_{\C}(q_0)=(\mu_0)_{\C}(\q_0)=\frac{1}{2}$. 

Induction: Assume that the statement holds for all words $w'$ with $w' <i$. Let $w=w' \sigma$ where $w' \in  \Sigma_{\B}^{i-1}$ and $\sigma\in \Sigma_{\B}$. By definition,  
$ \C^{w}_{i}(q_1) = \sum_{q \in Q_{\C}} \C^{w}_{i-1}(q) \cdot \delta_{\C}(q,\sigma)(q_1)$ for $q_1\in Q\setminus\{q_f\}$. By  induction hypothesis, since $\C^{w'}_{i-1}(q_2)=\C^{w'}_{i-1}(\q_2)$ for $q_2\in Q\setminus\{q_f\}$, and since $\delta_{\C}(q_2,\sigma)(q_1)=\delta_{\C}(\q_2,\sigma)(q_1)$, we conclude  
$ \C^{w}_{i}(q_1) = 2 \sum_{q_2 \in Q} \C^{w}_{i-1}(q_2) \cdot \delta_{\C}(q_2,\sigma)(q_1)=
2 \sum_{q_2 \in Q} (\frac{1}{2} \A^{w}_{i-1}(q_2) )\cdot (\frac{1}{2} \delta(q_2,\sigma)(q_1))=\frac{1}{2} \sum_{q_2 \in Q} \A^{w}_{i-1}(q_2) \cdot \delta(q_2,\sigma)(q_1)=\frac{1}{2} \A^{w}_{i}(q_1)$. Similarly, we obtain $ \C^{w}_{i}(\q_1)=\frac{1}{2}\A^{w}_{i}(q_1)$.

Now, we show that $\LW(\C)\neq \emptyset$ iff the PA $\A$ has value 1.
 
$\Leftarrow$ First, we assume that $\A$ has value 1. So, for all $\epsilon >0$, there exists a finite word $w \in \Sigma^*$ such that $P_{\A}(w)>1-\epsilon$.
Let $w_i$ be such that $P_{\A}(w_i)>1-2^{-i}$ for all $i \geq 1$.
We claim that the infinite word $v=(w_{i}\#)_{i \in \nat}=w_{1} \# w_{2} \# w_{3} \#\cdots$ is a weakly synchronizing word for $\C$. 
Let $\C^{v}_0 \C^{v}_1 \cdots$ be the outcome of $\C$ on~$v$. For $i> 1$, let $v_i=w_{1} \# \cdots \#   w_{i}$ be the prefix of $v$ which ends with $w_{i}$. Let 
 $\abs{v_i}$ denote the length of $v_i$, then 
$\abs{v_i}=(i-1)+\sum_{j=1}^{i} \abs{w_{j}}.$
The last $\#$ of the prefix $v_{i}$  is located at $\abs{v_{i-1}}+1$ (for $i>1$).
Thus,  $\C^{v}_{\abs{v_{i-1}}+1}=(\mu_0)_{\C}$  which is  the initial distribution of $\C$.
By construction, 
 $\C^{v}_{\abs{v_{i-1}}+1+\abs{w_i}}(q_f)= P_{\A}(w_i)$. 
Hence,  $\norm{\C^{v}_{\abs{v_{i}}}}> 1- 2^{-i}$ and $\lim_{i \to \infty}\norm{\C^{v}_{\abs{v_{i}}}}=1$. It implies that $\limsup_{n \to \infty}\norm{\C^{v}_{n}}=1$. Hence  $v$ is a weakly synchronizing word for $\C$, thus $\LW(\C)\neq \emptyset$.

$\Rightarrow$ Now let us assume that $\LW(\C)\neq \emptyset$. 
So, by definition there exists an infinite word $v=\sigma_0 \sigma_1 \sigma_2 \cdots$ such that $\limsup _{n \to \infty}\norm{\C^{v}_n}=1$.
We claim that $v$ contains infinitely many $\#$. 
Towards contradiction, assume that there exists $n \in \nat$ such that  for all $i>n$, $\sigma_i \neq \#$.

 Let $\C^{v}_{0}  \C^{v}_{1} \cdots$ be the outcome of $\C$ on  $v$. 
There are two cases: 
\begin{itemize}
\item There exists $j > n$ such that $\sigma_j=\dollar$. Since $\Post_{\C}(Q_{\C}, \dollar)=\{q_f,q_n,\q_n\}$ and $\Post_{\C}(\{q_f,q_n,\q_n\}, \Sigma )=\{q_n,\q_n\}$, we have $\C^{v}_{j+1+i}(q_n)=\C^{v}_{j+1+i}(\q_n)=\frac{1}{2}$ for all $i \in \nat$, a contradiction with the fact that $v$ is weakly synchronizing. 
\item For all $j >n$ we have $\sigma_j \neq \dollar$. 
Let $k$ be one  position after  the last $\#$ in~$v$, or $k=0$ if there is no $\#$. 
For all $i\in \nat$, let $v_i=\sigma_{k} \sigma_{k+1} \cdots \sigma_{k+i-1}$ be the subword starting at $k$ with length $i$; and let $\C^{v_i}_{0}  \C^{v_i}_{1} \cdots \C^{v_i}_{i}$ be the computation of $\C$  on the finite word $v_i$.
By Properties $P_1$ and $P_2$, $\C^{v_i}_{i}(q)=\C^{v_i}_{i}(\q)=\C^{v}_{k+i}(q)=\C^{v}_{k+i}(\q)$ for all $i\in \nat$ and all $q\in Q\setminus\{q_f\}$ and $\C^{v}_{k+i}(q_f)=0$. 
This gives $\norm{\C^{v}_{n+i}} \leq \frac{1}{2}$, a contradiction. 
\end{itemize}

\indent Now, we claim that $v$ contains infinitely many $\dollar$ too. 
Towards contradiction, assume that there exists $n \in \nat$ such that for all $i>n$, $\sigma_i \neq \dollar$. 
Let  $k_0 k_1 k_2 \cdots$ be the sequence of all positions in $v$ after $n$ where $\sigma_{k_j}=\#$ ($j\in \nat$) in increasing order ($k_{j+1}>k_j$). 
So, $k_j>n$ and $\sigma_i \neq \#$ for all $j \in \nat$ and all $i\neq k_j$. 
Let $v_i=\sigma_{k_j}\sigma_{k_j+1} \cdots \sigma_{i}$ be the finite subword of $v$ between the positions $k_j$ and $i$ where  $k_j\leq i<k_{j+1}$ for some $j\in\nat$. We see that $v_i=\# w$ for some $w\in\Sigma_{\B}^{*}$. 
Let $\C^{v_i}_{0}  \cdots \C^{v_i}_{i-k_j}$ be the computation of $\C$  on $v_i$.
By Properties $P_1$ and $P_2$, $\C^{v_i}_{i-k_j}(q)=\C^{v_i}_{i-k_j}(\q)=\C^{v}_{i}(q)=\C^{v}_{i}(\q)$ for all $i>n$ and all $q\in Q\setminus \{q_f\}$, and also $\C^{v_i}_{i}(q_f)=0$. 
This gives $\norm{\C^{v}_{i}} \leq \frac{1}{2}$ for all $i>n$, a contradiction with the fact that $v$ is weakly synchronizing. 

We showed that  $v$ contains infinitely many $\#$ and $\dollar$. Since $v$ is weakly synchronizing, for all $\epsilon>0$ there exists $m>0$ where $\norm{\C^{v}_m}>1-\epsilon$ (In fact, since $v$ is weakly synchronizing, for all $\epsilon>0$ for all $n>0$ there exists $m>n$ where $\norm{\C^{v}_m}>1-\epsilon$, but we do not need $n$ here).
For fixed $\epsilon < \frac {1}{2}$, let $m>0$ be such that $\norm{\C^{v}_m}>1-\epsilon$. 
Let $k<m$ be one  position after the last $\#$ before $\sigma_m$ (i.e., $\sigma_{k-1}=\#$), or  $k=0$ if there is no $\#$ before $\sigma_m$.
Let $w$ be  the finite subword of $v$ starting from position $k$ to the position $m$. 
By Properties $P_1$ and $P_2$, $\C^{v}_{m-k}(q)=\C^{v}_{m-k}(\q) \leq \frac{1}{2}$ for all $q\in Q\setminus \{q_f\}$. 
Therefore, since 
$\norm{\C^{v}_m}>1-\epsilon$, we have $\C^{v}_m(q_f)>1-\epsilon$ and $P_{\A}(w)>1-\epsilon$. Hence $\A$ has value~1. 
 
\begin{figure}[t]
\begin{center}

\begin{picture}(115,40)(0,2)

\node[Nmarks=n, Nw=33, Nh=24, dash={0.2 0.5}0](pa1)(13,20){}
\node[Nmarks=n, Nw=54, Nh=35, dash={0.4 1}0](pa)(22,20){}
\node[Nframe=n](label)(3,35){PA $\A$}

\node[Nmarks=n, Nw=33, Nh=24, dash={0.2 0.5}0](pa1)(82,20){}
\node[Nmarks=n, Nw=54, Nh=37, dash={0.4 1}0](pa)(91,20){}
\node[Nframe=n](label)(73,35){PA $\C$}

\node[Nframe=n](arrow)(55,20){{\Large $\Rightarrow$}}

\node[Nmarks=i](n3)(1,21){$q_0$}
\node[Nmarks=r,dash={2 0.5}0](n4)(25,26){}
\node[Nmarks=n](n5)(25,14){}
\node[Nmarks=r](qf)(42,26){$q_f$}
\node[Nmarks=n](qn)(42,14){$q_n$} 

\drawloop[ELside=l,loopCW=y, loopangle=-90, loopdiam=4](qn){$\Sigma$}
\drawedge(qf,qn){$\Sigma$}
\drawedge(n4,qf){$\dollar$}
\drawedge(n5,qn){$\dollar$}

\node[Nmarks=i, Nw=5, Nh=9](n3a)(70,20){$\ell_0$}
\node[Nmarks=r, Nw=5, Nh=9, dash={2 0.5}0](n4a)(92,26){}
\node[Nmarks=n, Nw=5, Nh=9](n5a)(92,14){}
\node[Nmarks=r](qfa)(108,26){$q_f$}
\node[Nmarks=n, Nw=5, Nh=9](qna)(108,14){$\ell_n$} 

\drawloop[ELside=l,loopCW=y, loopangle=-90, loopdiam=4](qna){$\Sigma$}
\drawedge(qfa,qna){$\Sigma$}
\drawedge(n4a,qfa){$\dollar$}
\drawedge(n5a,qna){$\dollar$}

\drawloop[ELside=l,loopCW=y, loopangle=90, loopdiam=4](n3a){$\#$}
\drawedge[ELpos=50, ELside=l, curvedepth=-2](n4a,n3a){$\#$}
\drawedge[ELpos=50, ELside=l, curvedepth=2](n5a,n3a){$\#$}
\drawedge[ELpos=50, ELside=l, curvedepth=-10](qfa,n3a){$\#$}
\drawedge[ELpos=50, ELside=l, curvedepth=10](qna,n3a){$\#$}
\end{picture}
\caption{The construction of the PA $\C$ from $\A$}\label{fig:undec}
\end{center}
\end{figure}
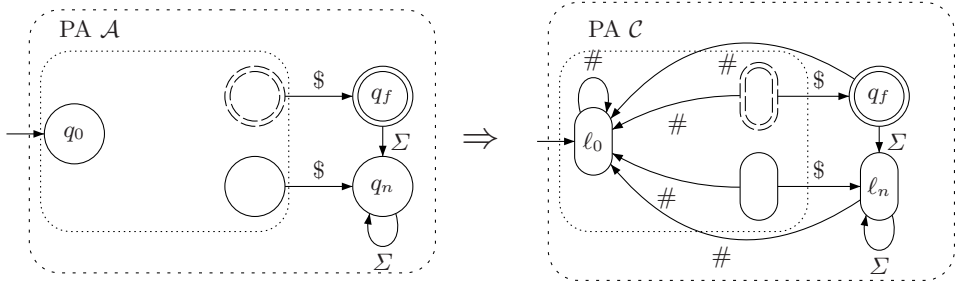
\end{proof}

\bibliographystyle{splncs03}
\bibliography{biblio}
\end{document}